\documentclass[runningheads]{llncs}
\usepackage[T1]{fontenc}
\usepackage{graphicx}
\usepackage{amssymb}
\usepackage{amsmath}
\usepackage[flushmargin]{footmisc}
\usepackage{bbding}

\begin{document}
\title{Truthful Two-Obnoxious-Facility Location Games with Optional Preferences and Minimum Distance Constraint}
\titlerunning{Truthful Two-Obnoxious-Facility Location Games}
%
\author{Xiaojia Han\inst{1} \and
Wenjing Liu\Envelope\inst{1,2} \and
Qizhi Fang\inst{1}}
\authorrunning{X. Han et al.}
%
\institute{School of Mathematical Sciences, Ocean University of China, Qingdao, 266100, China \and
Laboratory of Marine Mathematics, Ocean University of China, Qingdao, 266100, China\\
\email{hxj1069@stu.ouc.edu.cn},
\email{\{liuwj,qfang\}@ouc.edu.cn}}
%
%


%
\maketitle              
\begin{abstract}
In this paper, we study a truthful two-obnoxious-facility location problem, in which each agent has a private location in \([0,1]\) and a public optional preference over two obnoxious facilities, and there is a minimum distance constraint \(d \in [0,1]\) between the two facilities.
Each agent wants to be as far away as possible from the facilities that affect her, and the utility of each agent is the total distance from her to these facilities. The goal is to decide how to place the facilities in \([0,1]\) so as to incentivize agents to report their private locations truthfully as well as maximize the social utility.
First, we consider the special setting where \(d=0\), that is, the two facilities can be located at any point in \([0,1]\). We propose a deterministic strategyproof mechanism with approximation ratio of at most 4 and a randomized strategyproof mechanism with approximation ratio of at most 2, respectively. Then we study the general setting where \(d \in [0,1]\). 
We propose a deterministic strategyproof mechanism with approximation ratio of at most 8 and a randomized strategyproof mechanism with approximation ratio of at most 4, respectively. Furthermore, we provide lower bounds of 2 and \(\frac{14}{13}\) on the approximation ratio for any deterministic and any randomized strategyproof mechanism, respectively.

\keywords{Facility location game \and Mechanism design \and Strategyproof  \and Optional preferences \and Minimum distance constraint.}
\end{abstract}
\section{Introduction}\label{sec1}
In the classic truthful facility location problems, there is a set of agents who have private locations on the real line, and there is a facility (for example, a park or a school) that will be located somewhere on that line. The government aims to design a mechanism to determine the facility location that approximately optimizes some social objective. Since the mechanism operates based on the information reported by agents, each self-interested and rational agent may strategically misreport her location to manipulate the facility location, thus reducing her own cost or increasing her own utility. 
Therefore, the government's goal is to design mechanisms that guarantee truthful reporting by agents (e.g., strategyproof) while approximately optimizing the social objective. 
 The inconsistency between individual and collective interests implies that, in mechanism design, strategyproof mechanisms generally cannot guarantee the optimization of social objectives. Therefore, Procaccia and Tennenholtz~\cite{r1} introduced the approximation ratio, which is the worst-case ratio (over all possible instances) of the social objective value obtained by the mechanism to the optimal social objective value\footnote{This definition applies to cost-minimization social objectives. For utility-maximization social objectives, the approximation ratio is defined as the worst-case ratio (over all instances) of the optimal social objective value to the social objective value obtained by the mechanism.}, to measure the performance of strategyproof mechanisms. 
 
 In this paper, we study two-obnoxious-facility location games with optional preferences and minimum distance constraint with respect to two facilities. As a stylized example of our setting, consider locating a bar and a prison along a linear street. Both facilities are generally perceived as obnoxious by residents: the bar due to noise pollution and the prison due to potential psychological distress. However, some residents may be indifferent to whether these facilities are located close to their houses. Critically, due to the concern that former inmates congregating in nearby bars may cause disturbances, these facilities must maintain a minimum separation distance.

\subsection{Our Results}\label{subsubsec1}
Assuming that a set of agents have private locations in \([0,1]\) and public optional preferences over two obnoxious facilities. Specifically, some agents prefer to be far from exactly one facility (with utility defined as the distance to that facility), while the others prefer to be far from both facilities (with utility defined as the sum of distances to both facilities). There is a minimum distance constraint \(d \in [0,1]\) between the two facilities.
Our goal is to design mechanisms that take the agents' reported locations and the public preferences as input, and decide how to place the two facilities, so that (a) incentivize agents to report their locations truthfully, and (b) (approximately) maximize the social utility (i.e., the sum of the utilities of all agents).

In Section 3, we study the two-obnoxious-facility-location problem with optional preferences but without any distance constraint (i.e. \(d=0\)). We propose a deterministic strategyproof mechanism (Mechanism 1) with approximation ratio of at most \(4\), and a randomized strategyproof mechanism (Mechanism 2) with approximation ratio of at most \(2\).

In Section 4, we study the two-obnoxious-facility-location problem with optional preferences and minimum distance constraint (i.e. \(d \in [0,1]\)).  We design a deterministic strategyproof mechanism (Mechanism 3) with approximation ratio of at most \(8\), which switches between two cases depending on the cardinalities of the sets of agents with different preferences, and a randomized strategyproof mechanism (Mechanism 4) with approximation ratio of at most 2. Furthermore, we derive lower bounds of 2 and \(\frac{14}{13}\) on the approximation ratio for any deterministic and randomized strategyproof mechanism, respectively.

Our results are summarized in Table 1, where UB and LB indicate the upper bound and the lower bound of the approximation ratio, respectively.
\begin{table}
\caption{A summary of our results.}
\begin{tabular}{|l|l|l|}
\hline
\textbf{Minimum distance constraint} & \textbf{Deterministic} & \textbf{Randomized} \\
\hline
Special setting & UB: 4 (Theorem 2) & UB: 2 (Theorem 3) \\
               & LB: 2 (Theorem 6) & LB: $1.077$ (Theorem 8) \\
General setting & UB: 8 (Theorem 5) & UB: 2 (Theorem 7) \\
                & LB: 2 (Theorem 6) & LB: $1.077$ (Theorem 8) \\
\hline
\end{tabular}
\end{table}
\subsection{Related Work}\label{subsubsec1}
Mechanism design without money for facility location games has been widely studied in recent years. Moulin~\cite{r2} presented a complete characterization of strategyproof mechanisms when agents are located on the line and have single-peaked preferences. Procaccia and Tennenholtz~\cite{r1} were the first to propose approximate mechanism design without money through facility location problems. Fotakis and Tzamos~\cite{r3} investigated deterministic strategyproof mechanisms for two-facility location games on the line and proved a tight lower bound of \(n-2\) for minimizing the social cost. Cheng et al.~\cite{r11} discussed obnoxious facility location games, where each agent hopes that the facility will be built as far away from her as possible. Since then, researchers have studied various facility location problems from different perspectives. In what follows, we focus on reviewing two research settings closely related to our study: (a) agents have optional preferences over facilities, and (b) there are distance constraints between facilities.

\textbf{Facility Location Games with Optional Preferences.}
Serafino and Ventre~\cite{r4} introduced heterogeneous two-facility location games with optional preferences where facilities are built in a set of feasible alternative locations, and each agent has a public location and a private preference on facilities, with her cost being the sum of the distances to her preferred facilities (referred to as sum-variant). This model was later extended by Chen et al.~\cite{r5} from a discrete set to the continuous line, and the cost of each agent is the minimum distance or the maximum distance to the facilities they are interested in (referred to as min-variant or max-variant, respectively). Li et al.~\cite{r6} improved upper bounds of the approximation ratio when the individual cost is min-variant. Kanellopoulos et al.~\cite{r7} improved upper and lower bounds of the approximation ratio when the individual cost is sum-variant.

When agents' locations are private information and preferences on facilities are public information, Zhao et al.~\cite{r8} studied the two-facility location problem on a given set of candidate locations, where the cost of each agent is max-variant. Subsequently, Lofti et al.~\cite{r9} proposed a new strategyproof mechanism that improves upon Zhao et al.~\cite{r8}'s approximation ratio. Kanellopoulos et al.~\cite{r10} considered deterministic strategyproof mechanisms, where the cost of agent is sum-variant.
The facilities mentioned above are all desirable facilities. Kanellopoulos and Voudouris~\cite{r12} studied the two-obnoxious-facility location problem with optional preferences where facilities are built in a set of feasible alternative locations.

\textbf{Facility Location Games with Distance Constraints.}
For the case where there is a maximum distance constraint between two facilities, Zou and Li~\cite{r13} studied the dual preference model where each agent considers a single facility either desirable or obnoxious, but different agents have potentially different opinions. Gai et al.~\cite{r17} designed strategyproof mechanisms for two homogeneous facilities.
For the case where there is a minimum distance constraint between two facilities, Xu et al.~\cite{r18} proposed deterministic strategyproof mechanisms for both desirable and obnoxious facilities. Wu et al.~\cite{r19} studied two homogeneous and heterogeneous obnoxious facilities on a circle. Duan et al.~\cite{r20} and Xu et al.~\cite{r21} studied the two-desirable-facility location problem with fractional preferences and optional preferences, respectively.

Kanellopoulos and Voudouris~\cite{r12}, Xu et al.~\cite{r18}, and Xu et al.~\cite{r21} are the most related to our work among all researches on the heterogeneous facility location problem. However, there are several main differences between our work: (1) we study two obnoxious facility location problem with optional preferences and minimum distance constraint, whereas Kanellopoulos and Voudouris~\cite{r12} only consider each agent has optional preference, and Xu et al.~\cite{r18} only consider the minimum distance constraint; (2) each agent has a private location and a public preference on facilities in our model while the locations is public information and preferences is private information in Xu et al.~\cite{r21}, and Xu et al.~\cite{r21} consider two desirable facilities.

\section{Preliminaries}\label{sec2}
Let $N=\{1,2,\ldots,n\}$ be a set of agents located in an interval $I=[0,1]$ and $F=\{F_1,F_2\}$ be a set of two obnoxious facilities to be located in \([0,1]\). Each agent \(i \in N\) has a \textit{private location} $x_i \in I$ and a public facility preference \(p_i=(p_{i1},p_{i2}) \in \{0,1\}^2\) over the two facilities $F_1$ and $F_2$ with $p_{i1}+p_{i2} \geq 1$: indicating whether $i$ is affected by facility $F_j$ $(p_{ij}=1)$ or not $(p_{ij}=0)$, which is referred to as \textit{optional preferences}. Let $N_j$ ($j \in [2]$) be the set of agents that are affected by facility $F_j$. To be concise, denote an instance as $\mathbf{c}=(\mathbf{x},\mathbf{p})$, where \(\mathbf{x} = (x_{1}, x_{2}, \cdots, x_{n})\) is the \textit{location profile} and \(\mathbf{p} = (p_{1},p_{2},\cdots,p_{n})\) is the \textit{preference profile}.

Suppose that there is a minimum distance constraint $d \in [0,1]$ with respect to the two facilities. That is, the facility location profile $\mathbf{y}= (y_1,y_2)$ ($y_j$ for $F_j$) must satisfy $|y_1-y_2| \geq d$. The social planner selects the locations of the two obnoxious facilities based on the locations reported by the agents and the public preferences.

A deterministic mechanism $f(\mathbf{c})=(y_1,y_2) \in [0, 1]^2$ is a function which maps the reported location profile $\mathbf{x}$ and the public preference profile \(\mathbf{p}\) to a facility location profile $\mathbf{y}$. A randomized mechanism $f$ is a function which maps the reported location profile $\mathbf{x}$ and the public preference profile \(\mathbf{p}\) to a probability distribution over $[0, 1]^2$.

Given a deterministic outcome $f(\mathbf{c})=(y_1,y_2)$, the utility of agent $i$ is defined as the sum of her distances to the facilities that affect her
$$u_i((x_i,p_i),f(\mathbf{c}))=\sum_{j \in [2]} p_{ij} \cdotp |x_i-y_j|.$$

If $f(\mathbf{c})$ is a probability distribution returned by a randomized mechanism $f$, then the utility of agent $i$ is defined as the expected total distance to all facilities that affect her
\[u_i((x_i,p_i),f(\mathbf{c}))=\mathbb{E}_{(y_1,y_2) \sim f(\mathbf{c})}[\sum_{j \in [2]} p_{ij} \cdotp |x_i-y_j|].\]

 The \textit{social utility} of a mechanism \(f\) with respect to an instance \(\mathbf{c}\) is defined as the total utility of all agents $$SU(f|\mathbf{c})=\sum_{i \in N} u_i((x_i,p_i),f(\mathbf{c})).$$

 We aim to design mechanisms so that (a) the agents are incentivized to truthfully report their private locations, and (b) the social utility objective is (approximately) optimized.

A mechanism $f$ is said to be \textit{strategyproof} in two-obnoxious-facility location games if no agent can increase her utility by misreporting her location, regardless of the others’ strategies. Formally, for every \(i \in N\), every location profile \(\mathbf{x} \in I^n\), and every \(x_i' \in I\), it holds that
$$u_i((x_i,p_i),f(\mathbf{x}, \mathbf{p})) \geq u_i((x_i,p_i),f((x_{i}',\mathbf{x}_{-i}),\mathbf{p})),$$
where $\mathbf{x}_{-i}=(x_1, \cdots ,x_{i-1},x_{i+1}, \cdots , x_n)$ is the location profile of \(N \setminus \{i\}\).

A mechanism $f$ is \textit{group strategyproof} in the two-obnoxious-facility location games if for any nonempty subset of agents, at least one of them cannot benefit if they misreport their locations jointly. Formally, for every nonempty set $S \subseteq N$ of agents, every location profile \( \mathbf{x} =( \mathbf{x}_S, \mathbf{x}_{-S}) \in I^n \), and every \( \mathbf{x}'_S \in I^{|S|} \), there exists \( i \in S \), satisfying
\[ u_i((x_i,p_i),f((\mathbf{x}_S, \mathbf{x}_{-S}),\mathbf{p})) \geq u_i((x_i,p_i),f((\mathbf{x}'_S, \mathbf{x}_{-S}),\mathbf{p})). \]

We denote $OPT(\mathbf{c})$ to be the optimal social utility with respect to instance $\mathbf{c}=(\mathbf{x},\mathbf{p})$.
The \textit{approximation ratio} of a mechanism $f$ under the social utility objective is the worst-case ratio (over all possible instances) between the optimal social utility and the (expected) $f$-value, that is
$$\sup_{\mathbf{c}} \frac{SU(OPT|\mathbf{c})}{SU(f|\mathbf{c})}.$$

In this paper, we are interested in deterministic and randomized (group) strategyproof mechanisms with small approximation ratios under the social utility objective.
First, we consider the special setting where \(d=0\), that is, the two facilities can be located at any point within \([0,1]\). Then we study the general setting where \(d \in [0,1]\).

\section{Special setting}\label{sec3}

In this section, we consider the special setting where there is no distance constraint between two facilities (i.e., \(d=0\)). We present a deterministic strategyproof mechanism with approximation ratio of at most 4 and a randomized group strategyproof mechanism with approximation ratio of at most 2, respectively. 
We provide lower bounds of 2 and \(\frac{14}{13}\) on the approximation ratio for any deterministic and randomized strategyproof mechanism, respectively, as detailed in Section 4.

\subsection{Deterministic mechanism}\label{subsec3.1}
In this subsection, we propose a deterministic strategyproof mechanism with approximation ratio of at most 4.\\

\noindent \textbf{Mechanism 1}: Given an instance \(\mathbf{c}=(\mathbf{x},\mathbf{p})\), the facility location profile \(\mathbf{y}=(y_1,y_2)\) is given as follows: if more than $\frac{|N_{1}|}{2}$ agents in $N_{1}$ are located in $[0,1/2]$, output $y_{1} = 1$; otherwise, output $y_{1} = 0.$ If more than $\frac{|N_{2}|}{2}$ agents in $N_{2}$ are located in $[0,1/2]$, output $y_{2} = 1$; otherwise, output $y_{2} = 0.$

\begin{theorem}\label{thm1}
    Mechanism 1 is strategyproof.
\end{theorem}

Due to space constraints, the proof of Theorem 1 can be found in the appendix.\footnote{Due to space constraints, all missing proofs can be found in the appendix.}

\begin{lemma}
    In the special setting, for any instance \(\mathbf{c}\), the optimal social utility cannot exceed \(n+|N_1 \cap N_2|\).
\end{lemma}

\begin{theorem}\label{thm2}
    The approximation ratio of Mechanism 1 is at most 4.
\end{theorem}
\begin{proof}
  Let \(f\) denote Mechanism 1. For any instance \(\mathbf{c}=(\mathbf{x},\mathbf{p})\), we now discuss four cases with respect to $f(\mathbf{c})= (y_{1},y_{2})$.

Case 1. $(y_{1},y_{2}) = (1,1)$. In this case, there are more than $\frac{|N_{1}|}{2}$ $N_{1}$-agents and more than $\frac{|N_{2}|}{2}$ $N_{2}$-agents located in $[0,1/2]$. Thus, the social utility of Mechanism 1 is
\begin{align*}
SU\left( (1,1) \mid \mathbf{c} \right) &= \sum_{i \in N} u_{i}((x_i,p_i),f(\mathbf{c})) = \sum_{i \in N_{1}} (1 - x_{i}) + \sum_{i \in N_{2}} (1 - x_{i}) \\
&\geq \frac{|N_{1}|}{2} \cdot \left( 1 - \frac{1}{2} \right) + \frac{|N_{2}|}{2} \cdot \left( 1 - \frac{1}{2} \right) \\
&= \frac{|N_{1}| + |N_{2}|}{4}  \\
&=\frac{1}{4}(n+|N_1 \cap N_2|).
\end{align*}

Combining Lemma 1, we have $\frac{SU(OPT|\mathbf{c})}{SU ( f \mid \mathbf{c})} \leq 4$.

Case 2. $(y_{1},y_{2}) = (0,0)$. This case is similar to Case 1.

Case 3. $(y_{1},y_{2}) = (1,0)$. In this case, there are more than $\frac{|N_{1}|}{2}$ $N_{1}$-agents located in \([0,\frac{1}{2}]\) and more than $\frac{|N_{2}|}{2}$ $N_{2}$-agents located in \([\frac{1}{2},1]\). Thus, the social utility of Mechanism 1 is

\begin{align*}
SU\left( (1,0) \mid \mathbf{c} \right) &= \sum_{i \in N} u_{i}((x_i,p_i),f(\mathbf{c})) = \sum_{i \in N_{1}} (1 - x_{i}) + \sum_{i \in N_{2}}  x_{i} \\
&\geq \frac{|N_{1}|}{2} \cdot \left( 1 - \frac{1}{2} \right) + \frac{|N_{2}|}{2} \cdot \frac{1}{2} \\
&= \frac{|N_{1}| + |N_{2}|}{4}  \\
&=\frac{1}{4}(n+|N_1 \cap N_2|).
\end{align*}

Combining Lemma 1, we have $\frac{SU(OPT|\mathbf{c})}{SU\left( f \mid \mathbf{c} \right)} \leq 4$.

Case 4. $(y_{1},y_{2}) = (0,1)$. This case is similar to Case 3.

Therefore, the approximation ratio of Mechanism 1 is at most 4.
\qed
\end{proof}

\subsection{Randomized mechanism}\label{subsec3.2}
In this subsection, we propose a randomized group strategyproof mechanism with approximation ratio of at most 2.\\

\noindent \textbf{Mechanism 2}.\quad Given an instance \(\mathbf{c}=(\mathbf{x},\mathbf{p})\), return the facility location profile $(0,0)$ ,$(0,1), (1,0), $ and $ (1,1)$ with probability of $\frac{1}{4}$, respectively. 

\begin{theorem}\label{thm3}
    Mechanism 2 is a group strategyproof mechanism with approximation ratio of at most \( 2\).
\end{theorem}
\begin{proof}
    Mechanism 2 is group strategyproof since it does not depend on the agents' reported locations.

Let \(f\) denote Mechanism 2. For any instance \(\mathbf{c}=(\mathbf{x},\mathbf{p})\), the social utility of Mechanism 2 is
\begin{align*}
    SU(f|\mathbf{c})&=\frac{1}{4} (\sum_{i \in N_1} (x_i+x_i+(1-x_i)+(1-x_i))+\sum_{i \in N_2} (x_i+(1-x_i)+x_i+(1-x_i)))\\
    &=\frac{1}{4} (\sum_{i \in N_1} 2+\sum_{i \in N_2}2)\\
    &=\frac{1}{2}(|N_1|+|N_2|)\\
    &=\frac{1}{2}(n+|N_1 \cap N_2|).
\end{align*}
Combining Lemma 1, the approximation ratio of Mechanism 2 is at most \( 2\).
\qed
\end{proof}

\section{General setting}\label{sec4}
In this section, we study the two-obnoxious-facility location problem with optional preferences and minimum distance constraint \(d \in [0,1]\) between the two facilities.
\subsection{Deterministic Mechanisms}\label{subsec4.1}
In the following, we first present a deterministic strategyproof mechanism with approximation ratio of at most 8. Furthermore, we provide a lower bound of 2 on the approximation ratio for any deterministic strategyproof mechanism.\\

\noindent \textbf{Mechanism 3}.\quad Given an instance \(\mathbf{c}=(\mathbf{x},\mathbf{p})\), let \(j^*=argmax_{j \in [2]} |N_j \setminus N_{3-j}|.\) The two facilities are located as follows:
\begin{itemize}
    \item \textbf{Case 1.} If \( |N_1 \cap N_2|\geq|N_{j^*} \setminus N_{3-j^*}| \), let $n_L$ and $n_R$ denote the number of agents in $N_1 \cap N_2$ located in $[0,\frac{1}{2}]$ and $(\frac{1}{2},1]$, respectively. If $n_L \geq n_R$, output $(y_{j^*},y_{3-j^*})=(1,1-d)$; otherwise, output $(y_{j^*},y_{3-j^*})=(0,d).$
    \item \textbf{Case 2.} If \(|N_{j^*} \setminus N_{3-j^*}|>|N_1 \cap N_2|\) , let $n_L'$ and $n_R'$ denote the number of agents in $N_{j^*}\setminus N_{3-j^*} $ located in $[0,\frac{1}{2}]$ and $(\frac{1}{2},1]$, respectively. If $n_L' \geq n_R'$, output $(y_{j^*},y_{3-j^*})=(1,0)$; otherwise, output $(y_{j^*},y_{3-j^*})=(0,1).$
\end{itemize}

\begin{theorem}\label{thm4}
    Mechanism 3 is strategyproof.
\end{theorem} 
\begin{proof}
For any instance \(\mathbf{c}=(\mathbf{x},\mathbf{p})\), without loss of generality, assume that \(j^*=1\). Note that the two cases in Mechanism 3 depend only on agent's preferences, which are public information. Thus, $N_1, N_2$ are public. Hence, no agent can induce the mechanism to shift from (Case 1) to (Case 2) and vice versa by misreporting. We will discuss each case separately.

Case 1. \(|N_1 \cap N_2|\geq|N_1 \setminus N_2|\). Without loss of generality, assume that $(y_1,y_2)=(1,1-d)$ is the facility location profile that Mechanism 3 outputs. Since all \( N_1 \cap N_2\)-agents located in \( [\frac{1}{2},1]\) and \( (N_1 \backslash N_2) \cup (N_2 \backslash N_1)\)-agents cannot change the locations of \(F_1\) and \(F_2\) by misreporting, they have no incentive to misreport. Consider agent \(i \in N_1 \cap N_2\) in \([0,\frac{1}{2}] \). For \(d \in [0,\frac{1}{2}]\),
    \[ u_i((x_i,p_i),(0,d))=x_i+|d-x_i|  \leq (1-x_i)+(1-d-x_i)= u_i((x_i,p_i),(1,1-d)), \]
    and for \(d \in (\frac{1}{2},1]\),
    \[ u_i((x_i,p_i),(0,d))=x_i+(d-x_i)  \leq (1-x_i)+|1-d-x_i|
    =u_i((x_i,p_i),(1,1-d)). \]
    
    Thus, any agent \(i \in N_1 \cap N_2\) located in \( [0,\frac{1}{2}] \) has no incentive to misreport.

Case 2. \(|N_1 \setminus N_2| >|N_1 \cap N_2|\). Without loss of generality, assume that $(y_1,y_2)=(1,0)$ is the facility location profile that Mechanism 3 outputs. Consider agent \(i \in N_1 \cap N_2\), we have
    \[u_i((x_i,p_i),(0,1))=x_i+(1-x_i) = u_i((x_i,p_i),(1,0)),\]
    thus all \(N_1 \cap N_2\)-agents have no incentive to misreport.

Agents in \(N_2 \setminus N_1\) cannot affect the output of Mechanism 3 in this case, and have no incentive to misreport.
    
    The utility of agent \(i \in N_1 \backslash N_2\) located in \( [0,\frac{1}{2}] \) is
    \[u_i((x_i,p_i),(0,1))=x_i  \leq 1-x_i=u_i((x_i,p_i),(1,0)),\]
    so she has no incentive to misreport. 
    For any agent \(i \in N_1 \backslash N_2\) located in \( (\frac{1}{2},1] \), misreporting can only possibly increase \(n_L'\), and have no way to change the output.
\qed
\end{proof}

\begin{lemma}
    In the general setting, for any instance \(\mathbf{c}=(\mathbf{x},\mathbf{p})\), the optimal social utility cannot exceed \(n+(1-d) \cdot |N_1 \cap N_2|\).
\end{lemma}

\begin{proposition}\label{prop1}
    For instances with \( |N_1 \cap N_2|\geq|N_{j^*} \setminus N_{3-j^*}| \), the worst-case ratio between the optimal social utility and the social utility of Mechanism 3 is at most $2(4-d)$ for $d \in [0,\frac{1}{2}]$ and at most $\frac{4-d}{d}$ for $d \in (\frac{1}{2},1]$.
\end{proposition}

\begin{proof}
Without loss of generality, assume that \(j^*=1\), and $\mathbf{y}=(1,1-d)$ is the facility location profile that Mechanism 3 outputs.

    (1) If $d \in [0,\frac{1}{2}]$, the social utility of Mechanism 3 is
\begin{align*}
    SU(\mathbf{y}|\mathbf{c})&=\sum_{i \in N} u_i((x_i,p_i),\mathbf{y})\\
    &=\sum_{i \in N_1 \cap N_2}(|1-x_i|+|1-d-x_i|)+\sum_{i \in N_1 \backslash N_2}|1-x_i|+\sum_{i \in N_2 \backslash N_1}|1-d-x_i|\\
    &\geq \sum_{i \in N_1 \cap N_2, x_i\in[0,\frac{1}{2}]}(1-d)+\sum_{i \in N_1 \cap N_2, x_i\in(\frac{1}{2},1]}d\\
    &\geq \frac{|N_1 \cap N_2|}{2}.
\end{align*}

 Note that $|N_1 \cap N_2| \geq \frac{n}{3}$, since $|N_1 \cap N_2|+|N_1 \backslash N_2|+|N_2 \backslash N_1|=n$, $|N_1| \geq |N_2|$, and \( |N_1 \cap N_2|\geq|N_1 \setminus N_2| \).

Combining Lemma 2, we have 
\[\frac{SU(OPT|\mathbf{c})}{SU(\mathbf{y}|\mathbf{c})} \leq 2 \cdot \frac{n+(1-d) \cdot |N_1 \cap N_2|}{|N_1 \cap N_2|} \leq 2(4-d) \in [7,8].\]

(2) If \(d \in (\frac{1}{2},1]\),  the social utility of Mechanism 3 is
\begin{align*}
    SU(\mathbf{y}|\mathbf{c})&=\sum_{i \in N} u_i((x_i,p_i),\mathbf{y})\\
    &=\sum_{i \in N_1 \cap N_2}(|1-x_i|+|1-d-x_i|)+\sum_{i \in N_1 \backslash N_2}|1-x_i|+\sum_{i \in N_2 \backslash N_1}|1-d-x_i|\\
    &\geq d \cdot |N_1 \cap N_2|.
\end{align*}

Combining Lemma 2, we have 
\[\frac{SU(OPT|\mathbf{c})}{SU((\mathbf{y}|\mathbf{c})} \leq \frac{n+(1-d) \cdot |N_1 \cap N_2|}{d \cdot |N_1 \cap N_2|} \leq \frac{4-d}{d} \in [3,7).\]
\qed
\end{proof}

\begin{proposition}\label{prop2}
    For instances with \( |N_{j^*} \setminus N_{3-j^*}| \geq|N_1 \cap N_2|\), the worst-case ratio between the optimal social utility and the social utility of Mechanism 3 is at most $8$. 
\end{proposition}
\begin{proof}
    Without loss of generality, assume that \(j^*=1\), and $\mathbf{y}=(1,0)$ is the facility location profile that Mechanism 3 outputs.

The social utility of Mechanism 3 is
\begin{align*}
    SU(\mathbf{y}|\mathbf{c})&=\sum_{i \in N} u_i((x_i,p_i),\mathbf{y})\\
    &=\sum_{i \in N_1 \setminus N_2 }(1-x_i)+\sum_{i \in N_2 \setminus N_1 }x_i + \sum_{i \in N_1 \cap N_2}((1-x_i)+x_i)\\
    & \geq \frac{|N_1 \setminus N_2|}{2} \cdot \frac{1}{2}+|N_1 \cap N_2|\\
    &= \frac{|N_1|}{4}+\frac{3}{4}|N_1 \cap N_2|.
\end{align*}
Combining Lemma 2, we have
\begin{align}
    \frac{SU(OPT|\mathbf{c})}{SU(\mathbf{y}|\mathbf{c})} \leq 4 \cdot \frac{n+(1-d) \cdot |N_1 \cap N_2|}{|N_1|+3|N_1 \cap N_2|}.
\end{align}

The expression on the right-hand side of Eq. (1) is monotonically decreasing with both \(|N_1|\) and \(|N_1 \cap N_2|\).
Note that \(|N_1| \geq n/2\), since \(|N_1 \setminus N_2| \geq |N_2 \setminus N_1|\). Therefore
\[\frac{SU(OPT|\mathbf{c})}{SU(\mathbf{y}|\mathbf{c})} \leq 8.\]
\qed
\end{proof}

Combining Proposition \(1\) and Proposition \(2\), we obtain the following result.

\begin{theorem}\label{thm5}
    The approximation ratio of Mechanism 3 is at most $8$.
\end{theorem}

\begin{theorem}\label{thm6}
     Under the social utility objective, the approximation ratio of any deterministic strategyproof mechanism is at least 2.
\end{theorem}

\subsection{Randomized Mechanisms}\label{subsec4.2}
In this subsection, we consider randomized mechanisms. First, by setting an appropriate probability distribution, a randomized group strategyproof mechanism with approximation ratio of at most 2 can be achieved. Then, we provide a lower bound of \(\frac{14}{13}\) on the approximation ratio of any randomized strategyproof mechanism.\\

\noindent \textbf{Mechanism 4}.\quad Given an instance \(\mathbf{c}=(\mathbf{x},\mathbf{p})\), return the facility location profile \((0,1)\) and \((1,0)\) with probability of \(\frac{1}{2}\), respectively.

\begin{theorem}\label{thm7}
    Mechanism 4 is a group strategyproof mechanism with approximation ratio of at most $2$.
\end{theorem}
\begin{proof}
    Mechanism 4 is group strategyproof since it does not depend on the agents' reported locations and the known preferences.

Let \(f\) denote Mechanism 4. For any instance \(\mathbf{c}=(\mathbf{x},\mathbf{p})\), the social utility of Mechanism 4 is
\begin{align*}
    SU(f|\mathbf{c})&=\frac{1}{2} (\sum_{i \in N_1}x_i+\sum_{i \in N_2}(1-x_i))+\frac{1}{2} (\sum_{i \in N_1}(1-x_i)+\sum_{i \in N_2}x_i)\\
    &=\frac{1}{2} (\sum_{i \in N_1}x_i+\sum_{i \in N_1}(1-x_i))+\frac{1}{2} (\sum_{i \in N_2}(1-x_i)+\sum_{i \in N_2}x_i)\\
    &=\frac{1}{2}(|N_1|+|N_2|)\\
    &= \frac{1}{2}(n +|N_1\cap N_2|).
\end{align*}

Combining Lemma 2, we have 
\begin{align}
    \frac{SU(OPT|\mathbf{c})}{SU(f|\mathbf{c})} \leq 2 \cdot \frac{n+(1-d) \cdot |N_1 \cap N_2|}{n+|N_1\cap N_2|}.
\end{align}

The expression on the right-hand side of Eq. (2) is monotonically decreasing with \(|N_1\cap N_2|\). Note that \(|N_1 \setminus N_2| \geq 0 \).
Therefore
\[\frac{SU(OPT|\mathbf{c})}{SU(f|\mathbf{c})} \leq 2.\]
\qed
\end{proof}

\begin{theorem}\label{thm8}
    Under the social utility objective, the approximation ratio of any randomized strategyproof mechanism is at least \( 14/13 \approx 1.077\).
\end{theorem}

\section{Conclusions and Future Work}\label{sec5}

In this paper, we studied mechanism design for the two-obnoxious-facility-location problem with optional preferences and minimum distance constraint. For the objective of maximizing the social utility, we designed deterministic and randomized strategyproof mechanisms with proven approximation ratios, and derived lower bounds of 2 and \(\frac{14}{13}\) on the approximation ratios for any deterministic and randomized strategyproof mechanism, respectively.

There are many other interesting future research directions. For example, narrow the gap between the upper and the lower bounds for deterministic and randomized mechanisms in our settings.
We can also extend to other preference settings where agents may have triple preferences or dual preferences.
In addition, it would be interesting to change the assumption about which information is public or private, for example, assuming locations are public information and preferences are private information.
Finally, it might make sense to consider other social objective functions, such as the sum of square distance~\cite{squares} or the egalitarian welfare~\cite{egalitarian}.

\vspace{1em}
\noindent  \textbf{Acknowledgments.}
   This research was supported in part by the National Natural Science Foundation of China (12201590, 12171444) and Natural Science Foundation of Shandong Province (ZR2024MA031).

\subsubsection{\discintname}
\begin{footnotesize}
   The authors declare that they have no conflict of interest.
\end{footnotesize}

%
%
%
%
\bibliographystyle{splncs04}
\bibliography{main}

\begin{thebibliography}{10}
\providecommand{\url}[1]{\texttt{#1}}
\providecommand{\urlprefix}{URL }
\providecommand{\doi}[1]{https://doi.org/#1}

\bibitem{r5}
Chen, Z., Fong, K.C., Li, M., Wang, K., Yuan, H., Zhang, Y.: Facility location games with optional preference. Theoretical Computer Science  \textbf{847},  185--197 (2020)

\bibitem{r11}
Cheng, Y., Yu, W., Zhang, G.: Strategy-proof approximation mechanisms for an obnoxious facility game on networks. Theoretical Computer Science  \textbf{497},  154--163 (2013)

\bibitem{r20}
Duan, L., Gong, Z., Li, M., Wang, C., Wu, X.: Mechanism design for facility location with fractional preferences and minimum distance. In: Computing and Combinatorics: 27th International Conference, COCOON 2021, Tainan, Taiwan, October 24--26, 2021, Proceedings 27. pp. 499--511. Springer (2021)

\bibitem{r3}
Fotakis, D., Tzamos, C.: On the power of deterministic mechanisms for facility location games. ACM Transactions on Economics and Computation (TEAC)  \textbf{2}(4),  1--37 (2014)

\bibitem{r17}
Gai, L., Qian, D., Wu, C.: Two-facility location games with distance requirement. In: International Workshop on Frontiers in Algorithmics. pp. 15--24. Springer (2022)

\bibitem{r12}
Kanellopoulos, P., Voudouris, A.A.: Constrained truthful obnoxious two-facility location with optional preferences. arXiv preprint arXiv:2410.16131  (2024)

\bibitem{r7}
Kanellopoulos, P., Voudouris, A.A., Zhang, R.: On discrete truthful heterogeneous two-facility location. SIAM Journal on Discrete Mathematics  \textbf{37}(2),  779--799 (2023)

\bibitem{r10}
Kanellopoulos, P., Voudouris, A.A., Zhang, R.: Truthful two-facility location with candidate locations. Theoretical Computer Science  \textbf{1024},  114913 (2025)

\bibitem{egalitarian}
Li, F., Plaxton, C.G., Sinha, V.B.: The obnoxious facility location game with dichotomous preferences. Theoretical Computer Science  \textbf{961},  113930 (2023)

\bibitem{r6}
Li, M., Lu, P., Yao, Y., Zhang, J.: Strategyproof mechanism for two heterogeneous facilities with constant approximation ratio. In: Proceedings of the Twenty-Ninth International Conference on International Joint Conferences on Artificial Intelligence. pp. 238--245 (2021)

\bibitem{r9}
Lotfi, M., Voudouris, A.A.: On truthful constrained heterogeneous facility location with max-variant cost. Operations Research Letters  \textbf{52},  107060 (2024)

\bibitem{r2}
Moulin, H.: On strategy-proofness and single peakedness. Public Choice  \textbf{35}(4),  437--455 (1980)

\bibitem{r1}
Procaccia, A.D., Tennenholtz, M.: Approximate mechanism design without money. ACM Transactions on Economics and Computation (TEAC)  \textbf{1}(4),  1--26 (2013)

\bibitem{r4}
Serafino, P., Ventre, C.: Heterogeneous facility location without money. Theoretical Computer Science  \textbf{636},  27--46 (2016)

\bibitem{r19}
Wu, X., Mei, L., Zhang, G.: Two-facility location games with a minimum distance requirement on a circle. In: Combinatorial Optimization and Applications: 15th International Conference, COCOA 2021, Tianjin, China, December 17--19, 2021, Proceedings 15. pp. 497--511. Springer (2021)

\bibitem{r18}
Xu, X., Li, B., Li, M., Duan, L.: Two-facility location games with minimum distance requirement. Journal of Artificial Intelligence Research  \textbf{70},  719--756 (2021)

\bibitem{r21}
Xu, X., Zhang, J., Xie, L.: Minmax for facility location game with optional preference under minimum distance requirement. Journal of Combinatorial Optimization  \textbf{46}(4), ~23 (2023)

\bibitem{squares}
Ye, D., Mei, L., Zhang, Y.: Strategy-proof mechanism for obnoxious facility location on a line. In: International computing and combinatorics conference. pp. 45--56. Springer (2015)

\bibitem{r8}
Zhao, Q., Liu, W., Nong, Q., Fang, Q.: Constrained heterogeneous facility location games with max-variant cost. Journal of Combinatorial Optimization  \textbf{45}(3), ~90 (2023)

\bibitem{r13}
Zou, S., Li, M.: Facility location games with dual preference. In: Proceedings of the 2015 international conference on autonomous agents and multiagent systems. pp. 615--623 (2015)

\end{thebibliography}

\clearpage

\appendix
\section{Appendix}

\textbf{Proof of Theorem 1.}
\begin{proof}
    For any instance \(\mathbf{c}=(\mathbf{x},\mathbf{p})\), observe that $N_{1}$ and $N_{2}$ only depend on the preferences of the agents, which are public information. We now discuss four cases with respect to the facility location profile $\mathbf{y} = (y_{1},y_{2})$ that Mechanism 1 outputs.

Case 1. $(y_{1},y_{2}) = (1,1)$. In this case, each agent \(i \in N\) located in $[0,1/2]$ has no incentive to misreport since her utility is maximized, whether her preference is $p_{i} = (1,0),\ (0,1)\ \text{or}\ (1,1)$. Clearly, agents located in $(1/2,1]$ cannot affect the locations of $F_{1}$ and $F_{2}$, and have no incentive to misreport.

   Case 2. $(y_{1},y_{2}) = (0,0)$. This case is similar to Case 1.

   Case 3. $(y_{1},y_{2}) = (1,0)$. We only consider the agents located in \([0,\frac{1}{2}]\). In this case, each agent \(i \in N_1 \setminus N_2\)  has no incentive to misreport since her utility is maximized. Agents in \(N_2\setminus N_1\) cannot affect the output of Mechanism 1, and have no incentive to misreport. If agent \( i \in N_1 \cap N_2\) misreports her location as $x_{i}' \in (\frac{1}{2},1]$, then the new facility location profile \(\mathbf{y}'=(1,0)\) or \((0,0)\). Since the utility of agent $i$ does not increase, she has no incentive to misreport. The analysis of agents located in \((\frac{1}{2},1]\) is similar to that of agents located in \([0,\frac{1}{2}]\).

    Case 4. $(y_{1},y_{2}) = (0,1)$. This case is similar to Case 3.
\qed
\end{proof}

\noindent\textbf{Proof of Lemma 1.}
\begin{proof}
   For any facility location \(\mathbf{y}=(y_1,y_2)\), the utility of agent \(i \in (N_1 \backslash N_2) \cup (N_2 \backslash N_1)\) is
\[u_i((x_i,p_i),\mathbf{y})=\sum_{j \in [2]} p_{ij} \cdotp |x_i-y_j| \leq 1 ,\] 
and the utility of agent \(i \in N_1 \cap N_2\) is
\[u_i((x_i,p_i),\mathbf{y})=\sum_{j \in [2]} p_{ij} \cdotp |x_i-y_j| = |y_1-x_i|+|y_2-x_i| \leq 2.\]
Thus, the optimal social utility is
\begin{align*}
    SU(OPT|\mathbf{c})&= \sum_{i \in N_1 \backslash N_2} u_i((x_i,p_i),\mathbf{y})+\sum_{i \in N_2 \backslash N_1} u_i((x_i,p_i),\mathbf{y})+\sum_{i \in N_1 \cap N_2} u_i((x_i,p_i),\mathbf{y})\\
    & \leq |N_1 \backslash N_2|+|N_2 \backslash N_1|+2|N_1 \cap N_2| \\
    &=n+|N_1 \cap N_2|.
\end{align*} 
\end{proof}

\noindent\textbf{Proof of Lemma 2.}
\begin{proof}
    For any facility location \(\mathbf{y}=(y_1,y_2)\), the utility of agent \(i \in (N_1 \setminus N_2) \cup (N_2 \setminus N_1)\) is
\[u_i((x_i,p_i),\mathbf{y})=|y_j-x_i| \leq 1,\]
and the utility of agent \(i \in N_1 \cap N_2\) is
$$u_i((x_i,p_i),\mathbf{y})=\sum_{j \in [2]} p_{ij} \cdotp |x_i-y_j| = |y_1-x_i|+|y_2-x_i|  \leq 2-d.$$
Thus, the optimal social utility is
$$SU(OPT|\mathbf{c}) \leq |N_1 \setminus N_2|+|N_2 \setminus N_1|+(2-d) \cdot |N_1 \cap N_2|=n+(1-d) \cdot |N_1 \cap N_2|.$$
\end{proof}

\noindent\textbf{Proof of Theorem 6.}
\begin{proof}
    Denote \(f\) as an arbitrary deterministic strategyproof mechanism. Consider an instance \(\mathbf{c}=(\mathbf{x},\mathbf{p})\) with two agents, where \(\mathbf{x} = (1/3, 2/3)\) and $\mathbf{p}=((1,0),(1,0))$. \( f(\mathbf{c})=(y_1,y_2) \) is the facility location profile that Mechanism \(f\) outputs. \(OPT(\mathbf{c})=\) \((0,1)\) or \((1,0)\), then \(SU(OPT|\mathbf{c})=1\). Three cases will be discussed.

Case 1. \( y_1 \in [\frac{1}{3}, \frac{2}{3}] \). In this case, the social utility of \(f\) is \( SU(f| (\mathbf{c})) =|y_1-x_1|+|y_1-x_2|= \frac{1}{3} \). It follows that \( \frac{SU(OPT|(\mathbf{c}))}{SU(f|(\mathbf{c}))} =3\).

Case 2. \( y_1 \in [0,\frac{1}{3}) \). Then the utility of agent 1 in \(\mathbf{c}\) is \( u_1((x_1,p_1),\mathbf{y}) =  \frac{1}{3}-y_1 \). Now, consider another instance \( \mathbf{c'}=(\mathbf{x'},\mathbf{p}) \), where $x_{1}'=0$ and \( x_2 = \frac{2}{3} \), with the facility location profile \( f(\mathbf{c'}) =(y_1',y_2') \). 
\(OPT(\mathbf{c'})=(1,0)\), then \(SU(OPT|\mathbf{c'})=\frac{4}{3}\).
By strategyproofness, \( u_1((x_1,p_1),f(\mathbf{c'}))=|y_1'-x_1| \leq \frac{1}{3}-y_1 \). Otherwise, agent 1 in \(\mathbf{c}\) can benefit by misreporting from \( x_1 \) to \( x_1' \). This implies that \( y_1 \leq y_1' \leq 2x_1-y_1 \). Since \( y_1 \in [0,\frac{1}{3}) \) and \( x_1 = \frac{1}{3} \), we have \( y_1' \in [0,\frac{2}{3}] \). It is easy to get that \( SU(f| (\mathbf{c'})) =|y_1'-x_1'|+|y_1'-x_2|= \frac{2}{3} \). Hence, \( \frac{SU(OPT|(\mathbf{c'}))}{SU(f|(\mathbf{c'}))} = 2\)

Case 3. \( y_1 \in (\frac{2}{3},1] \). This case is similar to Case 2.
\qed
\end{proof}

\noindent\textbf{Proof of Theorem 8.}
\begin{proof}
    Denote \(f\) as an arbitrary randomized strategyproof mechanism. Consider an instance \(\mathbf{c}=(\mathbf{x},\mathbf{p})\) with two agents, where \(\mathbf{x} = (1/6, 5/6)\) and $\mathbf{p}=((1,0),(1,0))$. Let \(Y=(Y_1,Y_2)\) be a random variable following probability distribution \(f(\mathbf{c})\).
\(OPT(\mathbf{c})=(0,1)\) or \((1,0)\), then \(SU(OPT|\mathbf{c}) = 1\). 

It is worth noting that \(SU(f|\mathbf{c}) = \mathbb{E}_{Y \sim f(\mathbf{c})}[|Y_1 - 1/6| + |Y_1 - 5/6|] \leq 1\). Without loss of generality, we assume \(\mathbb{E}_{Y \sim f(\mathbf{c})}[|Y_1 - 5/6|] \leq 1/2\).

Consider another instance \(\mathbf{c'}=(\mathbf{x'},\mathbf{p})\), where $x_{1}'=1/6$ and \( x_2 = 1 \). \(OPT(\mathbf{c'})=(0,1)\), then we have \( SU(OPT|\mathbf{c'}) = 7/6 \). Let \(Y'=(Y_1',Y_2')\) be a random variable following probability distribution \(f(\mathbf{c'})\). Due to strategyproofness, we obtain \(\mathbb{E}_{Y' \sim f(\mathbf{c'})}[|Y_{1}' - 5/6|] \leq \mathbb{E}_{Y \sim f(\mathbf{c})}[|Y_1 - 5/6|] \leq 1/2\).

Denote \(Pr[Y_{1}' < 1/6] = q\), then
\begin{align*}
    \frac{2}{3}q &\leq \mathbb{E}_{Y' \sim f(\mathbf{c'})}[|Y_{1}' - 5/6| \mid Y_{1}' < 1/6] \cdot q + \mathbb{E}_{Y' \sim f(\mathbf{c'})}[|Y_{1}' - 5/6| \mid Y_{1}' \geq 1/6] \cdot (1 - q)\\
    &= \mathbb{E}_{Y' \sim f(\mathbf{c'})}[|Y_{1}' - 5/6|]  \leq 1/2,
\end{align*}
which implies, \(q \leq \frac{3}{4}\).

Note that \( \mathbb{E}_{Y' \sim f(\mathbf{c'})}[SU(f| (\mathbf{x'},\mathbf{p})) \mid Y_{1}' \geq 1/6] = 5/6 \). Hence, we have
\begin{align*}
    &SU(f|( \mathbf{x'},\mathbf{p}))\\ &= \mathbb{E}_{Y' \sim f(\mathbf{c'})}[SU(Y_{1}', \mathbf{x'}) \mid Y_{1}' < 1/6] \cdot q + \mathbb{E}_{Y' \sim f(\mathbf{c'})}[SU(Y_1', \mathbf{x'}) \mid Y_{1}'\geq 1/6] \cdot (1 - q) \\
    &\leq \frac{7}{6}q + \frac{5}{6}(1 - q).
\end{align*}
Therefore, the approximation ratio of \(f\) is at least

\[
\frac{7/6}{\frac{7}{6}q + \frac{5}{6}(1 - q)} = \frac{7}{5 + 2q} \geq \frac{14}{13} \approx 1.077.
\]
\qed
\end{proof}

\end{document}